\documentclass[journal]{IEEEtran}

\IEEEoverridecommandlockouts                              

\ifCLASSINFOpdf
  \usepackage[pdftex]{graphicx}
  \graphicspath{{../pdf/}{../jpeg/}}
   \DeclareGraphicsExtensions{.pdf,.jpeg,.png}
\else
  \usepackage[dvips]{graphicx}
   \graphicspath{{../eps/}}
   \DeclareGraphicsExtensions{.eps}
\fi

\usepackage{epsfig} 

\usepackage{setspace}
\usepackage{amsmath,amssymb,array,graphicx,verbatim}

\usepackage{amsthm}

\usepackage{enumerate}
\usepackage[shortlabels]{enumitem}

\newtheorem{theorem}{Theorem}
\newtheorem{lemma}{Lemma}

\newtheorem{corollary}{Corollary}

\newtheorem{assumption}{Assumption}
\newtheorem{remark}{Remark}
\usepackage{enumerate}
\usepackage{epstopdf}
\usepackage[noadjust]{cite}
  
\usepackage{dblfloatfix}
\usepackage{dsfont}
\usepackage{mathrsfs}

\usepackage{url}
\usepackage{algorithm}
\usepackage[noend]{algcompatible}
\usepackage{subcaption}
\usepackage{breqn}

\usepackage{tikz}
\usepackage{circuitikz}
\usepackage{tikz-cd}
\usetikzlibrary{arrows,shapes,calc,positioning}
\tikzstyle{block} = [draw,rectangle, rounded corners, minimum width=1cm, minimum height=0.8cm,text centered, line width=2pt ]
\tikzstyle{arrow} = [thick,->,>=stealth,line width=2pt]

\tikzset{cross/.style={cross out, draw=black, minimum size=2*(#1-\pgflinewidth), inner sep=0pt, outer sep=0pt},
cross/.default={1pt}}

\usepackage{enumerate}
\tikzset{
  shift left/.style ={commutative diagrams/shift left={#1}},
  shift right/.style={commutative diagrams/shift right={#1}}
}
\usetikzlibrary{calc}

\newcommand\rsmraise[1]{%
  \ifx#1\displaystyle .8\else
    \ifx#1\textstyle .8\else
      \ifx#1\scriptstyle .6\else
        .45%
      \fi
    \fi
  \fi}

\usepackage{epstopdf}
\usetikzlibrary{shapes}
\tikzstyle{block} = [draw,rectangle, rounded corners, minimum width=1cm, minimum height=0.8cm,text centered, line width=2pt ]
\tikzstyle{arrow} = [thick,->,>=stealth,line width=2pt]

\usetikzlibrary{arrows,decorations.markings,decorations}
\tikzset{
    addarrow/.style={decoration={markings, mark=at position 1 with {\arrow{stealth}}},
                     postaction={decorate}}
}

\newcommand{\R}{\mathbb{R}}			

\graphicspath{{./Sim_results/}}

\DeclareMathOperator{\ee}{\mathbb{E}}			
\DeclareMathOperator{\tr}{\mathbf{tr}}			

\title{
Learning-based Control of Unknown Linear Systems with Thompson Sampling
}

\author{Yi Ouyang, Mukul Gagrani and Rahul Jain
\thanks{Preliminary version of this paper  has been invited to the Allerton 2017 conference.}
\thanks{Y. Ouyang is currently with the University of California, Berkeley. M. Gagrani, and R. Jain are with the Department of Electrical
Engineering, University of Southern California, Los Angeles, CA. 
Email: {\tt ouyangyi@berkeley.edu, (mgagrani, rahul.jain)@usc.edu}}
}

\begin{document}

\maketitle

\begin{abstract}
We propose a Thompson sampling-based learning algorithm for the Linear Quadratic (LQ) control problem with unknown system parameters. 
The algorithm is called Thompson sampling with dynamic episodes (TSDE) where two stopping criteria determine the lengths of the dynamic episodes in Thompson sampling. The first stopping criterion controls the growth rate of episode length.
The second stopping criterion is triggered when the determinant of the sample covariance matrix is less than half of the previous value. 
We show under some conditions on the prior distribution that the expected (Bayesian) regret of TSDE accumulated up to time $T$ is bounded by $\tilde O(\sqrt{T})$. Here $\tilde O(\cdot)$ hides constants and logarithmic factors. This is the first $\tilde O(\sqrt{T})$ bound on expected regret of learning in LQ control. By introducing a re-initialization schedule, we also show that the algorithm is robust to time-varying drift in model parameters. Numerical simulations are provided to illustrate the performance of TSDE.
\end{abstract}


\section{Introduction}\label{sec:intro}


That the model and its parameters are known precisely is a pervasive assumption. And yet, for real-world systems, this is hardly the case. Typically, a set is known in which the model parameters lie. Furthermore, for many problems, we do not have the luxury of first performing system identification, and then using that in designing controllers. Learning model parameters and the corresponding optimal controller must be performed simultaneously at the fastest possible non-asymptotic rate. Classical adaptive control \cite{astrom1994adaptive,sastry1989adaptive,narendra1989stable} mostly provides asymptotic guarantees for non-stochastic systems. Results on Stochastic Adaptive Control are rather sparse. But recent advances in Online Learning \cite{cesa2006prediction} opens the possibility of using Online Learning-based methods for finding the optimal controllers to unknown stochastic systems. 

In this paper, we consider a linear stochastic system with quadratic cost (an LQ system) with unknown parameters. If the true parameters are known, then the problem is the classic stochastic LQ control where optimal control is a linear function of the state.  In the learning problem, however, the true system dynamics are unknown. This problem is also known as the adaptive control problem \cite{goodwin2014adaptive,kumar2015stochastic}. 

The early works in the adaptive control literature made use of the \emph{certainty equivalence principle}. 
The idea is to estimate the parameters from collected data and apply the optimal control by taking the estimates to be the true parameters. 
It was shown that the certainty equivalence principle may lead to the convergence of the estimated parameters to incorrect values \cite{becker1985adaptive} and thus results in suboptimal performance. 
This issue arises fundamentally from the lack of exploration. The controller must explore the environment to learn the system dynamics but at the same time it also needs to exploit the information available to minimize the accumulated cost. This leads to the well known exploitation-exploration trade-off in learning problems.

One approach to actively explore the environment is to add perturbations to the controls (see, for examples, \cite{chen1986convergence,guo1996self}).
However, the persistence of perturbations lead to suboptimal performance except in the asymptotic region.
To overcome this issue, Campi and Kumar \cite{campi1998adaptive} proposed a \text{cost-biased maximum likelihood} algorithm and proved its asymptotic optimality. More recent works \cite{abbasi2011regret,ibrahimi2012efficient} show a connection between the \text{cost-biased maximum likelihood} and the 
\emph{optimism in the face of uncertainty} (OFU) principle \cite{lai1985asymptotically} in online learning. The OFU principle handles the exploitation-exploration trade-off by making use of optimistic parameters.
Based on the OFU principle, \cite{abbasi2011regret,ibrahimi2012efficient} design algorithms that achieve $\tilde O(\sqrt{T})$ bounds on regret accumulated up to time $T$ with high probability. Here $\tilde O(\cdot)$ hides the constants and logarithmic factors. 
This regret scaling is believed to be optimal except for logarithmic factors because the similar linear bandit problem possesses a lower bound of $O(\sqrt{T})$ \cite{dani2008stochastic}.

One drawback of the OFU-based algorithms is their computational requirements. 
Each step of an OFU-based algorithm requires optimistic parameters as the solution of an optimization problem. Solving the optimization is computationally expensive. In recent years, Thompson sampling (TS) has become a popular alternative to OFU due to its computational simplicity (see \cite{russo2017tutorial} for a recent tutorial). 
It has been successfully applied to multi-armed bandit problems \cite{scott2010modern,chapelle2011empirical,agrawal2012analysis,kaufmann2012thompson,agrawal2013thompson} as well as to Markov Decision Processes (MDPs) \cite{osband2013more,abbasi2015bayesian,gopalan2015thompson}.
The idea dates back to 1933 due to Thompson \cite{thompson1933likelihood}. TS-based algorithms generally proceed in episodes. At the beginning of each episode, parameters are randomly sampled from the posterior distribution maintained by the algorithm. Optimal control is applied according to the sampled parameters until the next episode begins. Without solving any optimization problem, TS-based algorithms are computationally more efficient than OFU-based algorithms.

The idea of TS has not been applied to learning in LQ control until very recently \cite{abbasi2015bayesian,abeille2017thompson}.
One key challenge to adapt TS to LQ control is to appropriately design the length of the episodes. 
Abbasi-Yadkori and Szepesvári \cite{abbasi2015bayesian} designed a dynamic episode schedule for TS-based on their OFU-based algorithm \cite{abbasi2011regret}. Their TS-based algorithm was claimed to have a $\tilde O(\sqrt{T})$ growth, but a mistake in the proof of their regret bound was pointed out by \cite{osband2016posterior}.
A modified dynamic episode schedule was proposed in \cite{abeille2017thompson}, but it suffers a $\tilde O(T^{\frac{2}{3}})$ regret that is worse than the target $\tilde O(\sqrt{T})$ scaling. A related recent paper is \cite{kim2017thompson} which proposes a TS-based learning algorithm for finite state and finite action space stochastic control problems that is asymptotically optimal. Our focus is on non-asymptotic performance of learning-based control algorithms for stochastic linear systems that of course have both uncountable state and action spaces which is much more challenging.

In this paper, we consider the LQ control problem under two scenarios: with stationary parameters and with time-varying parameters. In the case of stationary parameters, we propose a Thompson sampling with dynamic episodes ({\tt TSDE}) learning algorithm. In {\tt TSDE}, there are two stopping criteria for an episode to end. The first stopping criterion controls the growth rate of episode length.
The second stopping criterion is the doubling trick similar to the ones in \cite{abbasi2011regret,abbasi2015bayesian,abeille2017thompson} that stops when the determinant of the sample covariance matrix becomes less than half of the previous value. 
Instead of a high probability bound on regret as derived in \cite{abbasi2011regret,ibrahimi2012efficient,abeille2017thompson},
we choose the expected (Bayesian) regret as the performance metric for the learning algorithm.
The reason is because in LQ control, a high probability bound does not provide a desired performance guarantee as the system cost may go unbounded in the bad event with small probability.
Under some conditions on the prior distribution, we show that the expected regret of {\tt TSDE}   accumulated up to time $T$ is bounded by $\tilde O(\sqrt{T})$.
In view of the mistake in \cite{abbasi2015bayesian}, our result would be  the first $\tilde O(\sqrt{T})$ bound on expected regret of learning in LQ control. When the system parameters are time-varying, we extend {\tt TSDE}   to the Time-Varying Thompson Sampling with Dynamic Episodes ({\tt TSDE-TV}) learning algorithm. Under some condition on the expected number of parameter jumps, we prove that {\tt TSDE-TV}   achieves sub-linear regret in $T$ which implies its asymptotical optimality under the average cost criterion. The performance of {\tt TSDE}   and {\tt TSDE-TV}   is also verified through numerical simulations. 

\section{Problem Formulation}\label{sec:problem}

\subsection{Preliminaries: Stochastic Linear Quadratic Control}

Consider a linear system controlled by a controller. The system dynamics are given by
\begin{align}
&x_{t+1}  = A x_t + B  u_t + w_t, \label{Model:system}
\end{align}
where $x_t \in \R^n$ is the state of the system plant, $u_t \in \R^m$ is the control action by the controller, and $w_t$ is the system noise which has the standard Gaussian distribution $N(0,I)$. 
$A$ and $B$ are system matrices with proper dimensions.
The initial state $x_1$ is assumed to be zero.

The control action $u_t = \pi_t(h_t)$
at time $t$ is a function $\pi_t$ of the history of observations $h_t = (x_{1:t}, u_{1:t-1})$ including states $x_{1:t}:=(x_1,\cdots,x_t)$ and controls $u_{1:t-1}=(u_1,\cdots,u_{t-1})$. We call $\pi = (\pi_1,\pi_2,\dots)$ a (adaptive) control policy. The control policy allows the possibility of randomization over control actions.

The cost incurred at time $t$ is a quadratic instantaneous function
\begin{align}
c_t = x_t^\top Q x_t + u_t^\top R u_t
\end{align}
where $Q$ and $R$ are positive definite matrices.

Let $\theta^\top = [A,B]$ be the system parameter including both the system matrices. 
Then $\theta \in \R^{d\times n}$ where $d = n+m$ with compact support $\Omega_1$.
When $\theta$ is perfectly known to the controller, minimizing the infinite horizon average cost per stage is a standard stochastic Linear Quadratic (LQ) control problem. 
Let $J(\theta)$ be the optimal per stage cost under $\theta$. That is,
\begin{align}
J(\theta) = \min_{\pi}\limsup_{T\rightarrow \infty}\frac{1}{T}\sum_{t=1}^{T} \ee^{\pi} [c_t | \theta]
\label{eq:opt_cost}
\end{align}
It is well-known that the optimal cost is given by
\begin{align}
J(\theta) = \tr(S(\theta))
\end{align}
if the following Riccati equation has a unique positive definite solution $S(\theta)$.
\begin{align}
&S(\theta) = Q + A^\top S(\theta) A  \notag\\
 & -A^\top S(\theta) B(R + B^\top S(\theta) B)^{-1} B^\top S(\theta) A.
 \label{eq:MRE}
\end{align}
Furthermore, for any $\theta$ and any $x$, the optimal cost function $J(\theta)$ satisfies the Bellman equation
\begin{align}
& J(\theta) +  x^\top S(\theta) x
=  
\min_{u} \Big\{ 
x^\top Q x +  u^\top R u 
\notag\\
 &\quad + \ee \Big[ x^\top_{t+1}(u) S(\theta) x_{t+1}(u)|x,\theta \Big]
 \Big\}
 \label{eq:DP_inf}
\end{align}
where $x_{t+1}(u) = \theta^\top [x^\top, u^\top]^\top + w_t$, and the optimal control that minimizes \eqref{eq:DP_inf} is equal to
\begin{align}
u = G(\theta) x
\label{eq:known_control}
\end{align}
with the gain matrix $G(\theta)$ given by
\begin{align}
G(\theta) = - (R + B^\top S(\theta) B)^{-1}B^\top S(\theta) A.
\end{align}

The problem we are interested in is the case when the system matrices $A,B$ are unknown.
When $\theta^\top = [A,B]$ is unknown, the problem becomes a reinforcement learning problem where the controller needs to learn the system parameter while minimizing the cost.

We first consider a learning problem with stationary parameters, and then with time-varying parameters.

\subsection{Reinforcement Learning with Stationary Parameter}

Consider the linear system 
\begin{align}
&x_{t+1}  = A_1 x_t + B_1  u_t + w_t, \label{Model:learning_stationary}
\end{align}
where $A_1$ and $B_1$ are fixed but unknown system matrices.
Let $\theta_1^\top = [A_1,B_1]$ be the model parameter.
We adopt a Bayesian setting and assume that there is a prior distribution $\mu_1$ for $\theta_1$.

Since the actual parameter $\theta_1$ is unknown, we define the expected regret of a policy $\pi$ compared with the optimal cost $J(\theta_1)$ to be
\begin{align}
R(T,\pi) = \ee\Big[ \sum_{t=1}^T \Big[ c_t  - J(\theta_1) \Big] \Big].
 \label{eq:regret_stationary}
\end{align}
The above expectation is with respect to the randomness for $W_t$, the prior distribution $\mu_1$ for $\theta_1$, and the randomized algorithm.
The learning objective is to find a control algorithm that minimizes the expected regret.

\subsection{Reinforcement Learning with Time-Varying Parameter}
Consider the time-varying system 
\begin{align}
&x_{t+1}  = A_t x_t + B_t  u_t + w_t, \label{Model:learning_stationary}
\end{align}
with system matrices $A_t$ and $B_t$.
The model parameter $\theta_t^\top = [A_t,B_t]$ is time-varying and unknown to the controller.

We assume that the parameter $(\theta_t,t=1,2,\dots)$ is a jump process. When it jumps, the new parameter is generated from the  prior distribution $\mu_1$.
We use $j_t \in \{0,1\},t=1,2,\dots$ to indicate the jumps. Then 
$\theta_{t} = \theta_{t-1}$ if $j_t = 0$, and $\theta_{t}$ is generated (independently of the past) from $\mu_1$ if $j_t = 1$.
The jump process $(j_t, t=1,2,\dots)$ is assumed to be independent of the system noise.

Since $J(\theta_t)$ is the optimal cost under $\theta_t$, we define the expected regret of a policy $\pi$ to be
\begin{align}
R_{TV}(T,\pi) = \ee\Big[ \sum_{t=1}^T \Big[ c_t  - J(\theta_t) \Big] \Big].
\label{eq:regret_TV}
\end{align}
The above expectation is with respect to the randomness for $W_t$, the distribution for the jump process $(\theta_t,t=1,2,\dots)$, and the randomized algorithm.
The learning objective is to find a control algorithm that minimizes the expected regret.

\section{Thompson Sampling Based Control Policies}\label{sec:algos}

In this section, we develop Thompson Sampling (TS)-based control policies for the problems with stationary  and time-varying parameters.

\subsection{Thompson Sampling for Stationary Parameter}

For the reinforcement learning problem with stationary parameters, we make the following assumption on the prior distribution $\mu_1$.
\begin{assumption}
\label{assum:prior}
The prior distribution $\mu_1$ consists of independent Gaussian distributions projected on a compact support $\Omega_1 \subset \R^{d\times n}$ such that for any $\theta \in \Omega_1$, the Riccati equation \eqref{eq:MRE} with $[A,B] = \theta^\top$ has a unique positive definite solution.  

Specifically, there exist $\hat \theta_1(i) \in \R^d$ for $i=1,\dots,n$ and a positive definite matrix $\Sigma_1 \in \R^{d\times d}$ such that for any $\theta \in \R^{d\times n}$
\begin{align}
&\mu_{1} = \bar \mu_1|_{\Omega_1}, \quad \bar\mu_1(\theta) = \prod_{i=1}^n \bar\mu_1(\theta(i))
\\
& \bar\mu_1(\theta(i)) \equiv N(\hat \theta_t(i),\Sigma_1) \text{ for } i=1,\dots,n.
\label{eq:mu1}
\end{align}
Here $\theta(i)$ denotes $\theta$'s $i$th column ($\theta = [\theta(1),\dots,\theta(n)]$).
\end{assumption}

Note that under the prior distribution, the mean $\hat \theta_1(i)$ for each column of $\theta_1$ may be different, but they have the same covariance matrix $\Sigma_1$.

At each time $t$, given the history of observations $h_t = (x_{1:t},u_{1:t-1})$, we define $\mu_t$ to be the posterior belief of $\theta_1$ given by
\begin{align}
\mu_t(\Theta) = \mathbb{P}(\theta_1 \in \Theta | h_t).
\end{align}
The posterior belief can be computed according to the following lemma.
\begin{lemma}
\label{lm:bfupdates}
The posterior belief $\mu_t$ on the parameter $\theta_1$ satisfies
\begin{align}
&\mu_{t} = \bar \mu_t|_{\Omega_1}
, \quad \bar\mu_t(\theta) = \prod_{i=1}^n \bar\mu_t(\theta(i))
\\
& \bar\mu_1(\theta(i)) \equiv N(\hat \theta_t(i),\Sigma_t)
\label{eq:mut}
\end{align}
where $\hat \theta_t(i), i=1,\dots,n,$ and $\Sigma_t$ can be sequentially updated using observations as follows.
\begin{align}
&\hat \theta_{t+1}(i) 
= \hat \theta_t(i) + \frac{\Sigma_t z_t (x_{t+1}(i) - \hat \theta_t(i)^\top z_t )}{1+z_t^\top\Sigma_tz_t}
\label{eq:theta_hat_update}
\\
&\Sigma_{t+1} = \Sigma_t -  \frac{\Sigma_t z_tz_t^\top \Sigma_t }{1+z_t^\top\Sigma_t  z_t} 
\label{eq:sigma_update}
\end{align}
where $z_t = [x_t^\top,u_t^\top]^\top \in \R^{n+m}$.
\end{lemma}
Lemma \ref{lm:bfupdates} can be proved using arguments for the least square estimator. For example, see \cite{sternby1977consistency} for a proof.

\begin{remark}
\label{rm:invcov}
Instead of the Kalman filter-type equation \eqref{eq:sigma_update}, $\Sigma_t$ can also be computed by
\begin{align}
\Sigma^{-1}_{t+1} = \Sigma^{-1}_t  + z_t z_t^\top.
\end{align}
\end{remark}

Let's introduce the Thompson Sampling with Dynamic Episodes ({\tt TSDE}) learning algorithm.
\begin{algorithm}[H]
\caption{{\tt TSDE}}
\begin{algorithmic}
\STATE Input: $\Omega_1, \hat \theta_1, \Sigma_1$ 
\STATE Initialization: $t\leftarrow 1$, $t_k \leftarrow 0$ 
\FOR{ episodes $k=1,2,...$}
	\STATE{$T_{k-1} \leftarrow t - t_{k}$}
	\STATE{$t_{k}  \leftarrow t$}
	\STATE{Generate $\tilde\theta_{k} \sim \mu_{t_k}$}
	\STATE{Compute $G_{k} = G(\tilde\theta_k)$ from \eqref{eq:DP_inf}-\eqref{eq:known_control}}

	\WHILE{$t \leq t_k+T_{k-1}$ and $\det(\Sigma_t) \geq 0.5\det(\Sigma_{t_k}) $}
	\STATE{Apply control $u_t = G_k x_t$}
	\STATE{Observe new state $x_{t+1}$}
	\STATE{Update $\mu_{t+1}$ according to \eqref{eq:theta_hat_update}-\eqref{eq:sigma_update}}
	\STATE{$t  \leftarrow t+1$}
	\ENDWHILE
\ENDFOR
\end{algorithmic}
\end{algorithm}

The {\tt TSDE} algorithm operates in episodes. 
Let $t_k$ be start time of the $k$th episode and $T_k = t_{k+1}-t_{k}$ be the length of the episode with the convention $T_0 = 1$.
From the description of the algorithm, $t_1 = 1$ and $t_{k+1}, k\geq 1,$ is given by
\begin{align}
t_{k+1} = \min\{ & t>t_{k}:\quad 
t > t_{k} + T_{k-1} 
\notag\\
&\quad \text{ or } 
\det(\Sigma_t) < 0.5\det(\Sigma_{t_k}) 
\}.
\label{eq:tk}
\end{align}

At the beginning of episode $k$, a parameter $\theta_k$ is sampled from the posterior distribution $\mu_{t_k}$. During each episode $k$, controls are generated by the optimal gain $G_k$ for the sampled parameter $\theta_k$. 
One important feature of {\tt TSDE} is that its episode lengths are not fixed. The length $T_k$ of each episode is dynamically determined according to two stopping criteria: (i) $t > t_k+T_{k-1}$, and (ii) 
$\det(\Sigma_t) < 0.5\det(\Sigma_{t_k}) $.
The first stopping criterion provides that the episode length grows at a linear rate without triggering the second criterion.
The second stopping criterion ensures that the determinant of sample covariance matrix during an episode should not be less than half of the determinant of sample covariance matrix at the beginning of this episode.

\subsection{Thompson Sampling for Time-Varying Parameter}

For the learning problem with stationary parameter, we assume that the prior distribution $\mu_1$ which generates the parameter after each jump satisfies Assumption \ref{assum:prior}.

We now introduce the Time-Varying Thompson Sampling with Dynamic Episodes ({\tt TSDE-TV}) learning algorithm.
%
\begin{algorithm}[H]
\caption{\tt TSDE-TV}
\begin{algorithmic}
\STATE Input: $\Omega_1, \hat \theta_1, \Sigma_1$ and a parameter $q$
\STATE Initialization: $t\leftarrow 1$, $t_k \leftarrow 0$,  $s_l \leftarrow 1$, $l \leftarrow 1$ 
\FOR{ episodes $k=1,2,...$}
		\STATE{$T_{k-1} \leftarrow t - t_{k}$}
	\STATE{$t_{k}  \leftarrow t$}
	\STATE{Generate $\tilde\theta_{k} \sim \mu_{t_k}$}
	\STATE{Compute $G_{k} = G(\tilde\theta_k)$ from \eqref{eq:DP_inf}-\eqref{eq:known_control}}

	\WHILE{$t \leq t_k+T_{k-1}$ and $\det(\Sigma_t) \geq 0.5\det(\Sigma_{t_k}) $}
	\IF {$t \geq s_l+l^q$}
		\STATE{Re-initialize: $t_{k} \leftarrow t-1$, $\hat\theta_t \leftarrow \hat\theta_1$,$\Sigma_t \leftarrow \Sigma_1$ }
		\STATE{$s_{l} \leftarrow t$, $l \leftarrow l+1$}
		\STATE{\textbf{break}}
	\ELSE	
	\STATE{Apply control $u_t = G_k x_t$}
	\STATE{Observe new state $x_{t+1}$}
	\STATE{Update $\mu_{t+1}$ according to \eqref{eq:theta_hat_update}-\eqref{eq:sigma_update}}
	\STATE{$t  \leftarrow t+1$}
	\ENDIF
	\ENDWHILE
\ENDFOR
\end{algorithmic}
\end{algorithm}

In {\tt TSDE-TV}, $s_l$ is the time when the algorithm re-initializes. 
The idea of {\tt TSDE-TV} is to re-initialize {\tt TSDE} to adapt to the jumps of the model parameter.
{\tt TSDE-TV} re-initializes if the time difference between the current episode and the previous re-initialization is long enough. The time difference between two re-initializations is $l^q$ which increases at a rate determined by the parameter $q$.

\section{Regret Analysis for Stationary Parameter}\label{sec:analysis}

In this section, we analyze the regret of {\tt TSDE} in the stationary parameter case. In the regret analysis, we make the following assumption on the prior distribution.
\begin{assumption}
\label{assum:stabilizability}
There exists a positive number $\delta <1$ such that for any $\theta \in \Omega_1$, we have
$\rho(A_1+ B_1 G(\theta)) \leq \delta$. Here $\rho(\cdot)$ is the spectral radius of a matrix, i.e. the largest absolute value of its eigenvalues.
\end{assumption}
This assumption ensures that the closed-loop system is stable under the learning algorithm. 
A weaker assumption in \cite{ibrahimi2012efficient} can ensure that Assumption \ref{assum:stabilizability} is satisfied for $\theta= \theta_k$ with high probability.

Since $J(\cdot)$, $S(\cdot)$, and $G(\cdot)$ are well-defined functions on the compact set $\Omega_1$, there exists finite numbers $M_{J}$, $M_{\theta}$, $M_S$, and $M_G$ such that
$J(\theta) \leq M_{J}$,
$||\theta|| \leq M_{\theta}$, $||S(\theta)|| \leq M_S$, and 
$||[I, G(\theta)^\top]|| \leq M_G$ for all $\theta \in \Omega_1$.

The main result of this section is the following bound on expected regret of {\tt TSDE} in the stationary parameter case.
\begin{theorem}
\label{thm:regret}
Under Assumptions \ref{assum:prior} and \ref{assum:stabilizability}, the expected regret \eqref{eq:regret_stationary} of {\tt TSDE} satisfies
\begin{align}
R(T,\text{{\tt TSDE}}) \leq \tilde O\Big(
\sqrt{T}
\Big)
\end{align}
where $\tilde O(\cdot)$ hides all constants and logarithmic factors.
\end{theorem}

To prove Theorem \ref{thm:regret}, we first provide bounds on the system state and the number of episodes.
Then, we give a decomposition for the expected regret and derive upper bounds for each term of the regret.

Let $X_T = \max_{t\leq T} \lVert x_{t} \Vert$ be the maximum value of the norm of the state and 
$K_T$ be the number of episodes over the horizon $T$. Then we have the following properties.

\begin{lemma}
\label{lm:bound_X}
For any $j\geq 1$ and any $T$ we have
\begin{align}
\ee \Big[ X_T^j \Big] \leq O\Big( \log(T)(1-\delta)^{-j}\Big).
\end{align}
\end{lemma}

\begin{lemma}
\label{lm:KT}
The number of episodes is bounded by
\begin{align}
K_T \leq 
O\left( \sqrt{ 2d T\log(T X_T^2)} \right).
\end{align}
\end{lemma}
The proofs of Lemmas \ref{lm:bound_X} and \ref{lm:KT} are in the appendix.

Following the steps in \cite{abbasi2011regret} using the Bellman equation \eqref{eq:DP_inf}, for $t_k \leq t < t_{k+1}$ during the $k$th episode, the cost of {\tt TSDE} satisfies
\begin{align}
c_t = & J(\tilde\theta_k) +  x^\top_t S(\tilde\theta_k) x_t
- \ee \Big[ x_{t+1}^\top S(\tilde\theta_k) x_{t+1}|x_t,\tilde\theta_k\Big]
\notag\\
& +(\theta_1^\top z_t)^\top S(\tilde\theta_k) \theta_1^\top z_t - (\tilde\theta_k^\top z_t)^\top S(\tilde\theta_k)\tilde\theta_k^\top z_t .
 \label{eq:DP_analysis3}
\end{align}

Then from \eqref{eq:DP_analysis3}, the expected regret of {\tt TSDE} can be decomposed into
\begin{align}
R(T,\text{{\tt TSDE}})
=& \ee\Big[ \sum_{k=1}^{K_T}\sum_{t=t_k}^{t_{k+1}-1} c_t  \Big] - T \ee\Big[J(\theta_1) \Big]
\notag\\
= & R_0+R_1 + R_2
\label{eq:regret_decomp}
\end{align}
where
\begin{align}
&R_0 = \ee\Big[ \sum_{k=1}^{K_T}T_k J(\tilde\theta_k) \Big]- T \ee\Big[J(\theta_1) \Big],
\\
&R_1 = \ee\Big[\sum_{k=1}^{K_T}\sum_{t=t_k}^{t_{k+1}-1}  \Big[x^\top_t S(\tilde\theta_k) x_t - x_{t+1}^\top S(\tilde\theta_k) x_{t+1}\Big]\Big],
\\
&R_2 
= \ee\Big[\sum_{k=1}^{K_T}\sum_{t=t_k}^{t_{k+1}-1}  
\notag\\
&\hspace{3em}
\Big[ (\theta_1^\top z_t)^\top S(\tilde\theta_k) \theta_1^\top z_t 
-(\tilde\theta_k^\top z_t)^\top S(\tilde\theta_k) \tilde\theta_k^\top z_t
\Big]\Big].
\end{align}

In the following, we proceed to derive bounds on $R_0$, $R_1$ and $R_2$.

As discussed in \cite{osband2013more,osband2016posterior,russo2014learning}, one key property of Thompson/Posterior Sampling algorithms is that for any function $f$, $\ee[f(\theta_t)] = \ee[f(\theta_1)]$ if $\theta_t$ is sampled from the posterior distribution at time $t$. 
However, our {\tt TSDE} algorithm has dynamic episodes that requires us to have the stopping-time version of the above property whose proof is in the appendix.
\begin{lemma}
\label{lm:PS_expectation}
Under {\tt TSDE}, $t_k$ is a stopping time for any episode $k$. Then for any measurable function $f$ and any $\sigma(h_{t_k})-$measurable random variable $X$, we have
\begin{align}
\ee \Big[ f(\tilde{\theta}_{k},X)\Big]= & \ee \Big[ f(\theta_1,X) \Big].
\end{align}
\end{lemma}

Based on the key property of Lemma \ref{lm:PS_expectation}, we establish an upper bound on $R_0$.
\begin{lemma}
\label{lm:R0}
The first term $R_0$ is bounded as
\begin{align}
R_0 \leq &M_J \ee[K_T].
\end{align}
\end{lemma}
\begin{proof}
From monotone convergence theorem, we have
\begin{align}
R_0 
= &\ee\Big[ \sum_{k=1}^{\infty}\mathds{1}_{\{t_{k}\leq T\}}T_k J(\tilde\theta_k) \Big]- T \ee\Big[J(\theta_1) \Big]
\notag\\
=& \sum_{k=1}^{\infty}\ee\Big[ \mathds{1}_{\{t_{k}\leq T\}}T_k J(\tilde\theta_{k})\Big]  - T\ee\Big[J(\theta_1) \Big].
\end{align}
Note that the first stopping criterion of {\tt TSDE} ensures that $T_k \leq T_{k-1}+1$ for all $k$.
Because $J(\tilde\theta_{k}) \geq 0$, each term in the first summation satisfies
\begin{align}
\ee\Big[ \mathds{1}_{\{t_{k}\leq T\}}T_k J(\tilde\theta_{k})\Big]
\leq &\ee\Big[ \mathds{1}_{\{t_{k}\leq T\}}(T_{k-1}+1) J(\tilde\theta_{k})\Big].
\end{align}
Note that $\mathds{1}_{\{t_{k}\leq T\}}(T_{k-1}+1)$ is measurable with respect to $\sigma(h_{t_k})$. Then, Lemma \ref{lm:PS_expectation} gives
\begin{align}
&\ee\Big[ \mathds{1}_{\{t_{k}\leq T\}}(T_{k-1}+1) J(\tilde\theta_{k})\Big]
\notag\\
= &\ee\Big[\mathds{1}_{\{t_{k}\leq T\}}(T_{k-1}+1) J(\theta_{1}) \Big].
\end{align}

Combining the above equations, we get
\begin{align}
R_0 
\leq &\sum_{k=1}^{\infty}\ee\Big[\mathds{1}_{\{t_{k}\leq T\}}(T_{k-1}+1) J(\theta_{1}) \Big] - T\ee\Big[J(\theta_1) \Big]
\notag\\
= &\ee\Big[\sum_{k=1}^{K_T}(T_{k-1}+1) J(\theta_{1}) \Big] - T\ee\Big[J(\theta_1) \Big]
\notag\\
= &\ee\Big[K_T J(\theta_{1})\Big]
+ \ee\Big[\Big(\sum_{k=1}^{K_T}T_{k-1} - T\Big)J(\theta_{1})\Big]
\notag\\
\leq & M_J\ee\Big[K_T\Big]
\end{align}
where the last equality holds because $J(\theta_{1}) \leq M_J$ and $\sum_{k=1}^{K_T}T_{k-1} \leq T$.
\end{proof}

The term $R_1$ can be upper bounded using $K_T$ and $X_T$.
\begin{lemma}
\label{lm:R1}
 The second term $R_1$ is bounded by
\begin{align}
R_1 \leq M_S\ee\Big[K_T X_T^2 \Big].
\end{align}
\end{lemma}
\begin{proof}
From the definition of $R_1$ we get
\begin{align}
R_1 = &\ee\Big[\sum_{k=1}^{K_T}\sum_{t=t_k}^{t_{k+1}-1}  \Big[x^\top_t S(\tilde\theta_k) x_t - x_{t+1}^\top S(\tilde\theta_k) x_{t+1}\Big]\Big]
\notag\\
=&\ee\Big[\sum_{k=1}^{K_T}\Big[x^\top_{t_k} S(\tilde\theta_k) x_{t_k} - x_{t_{k+1}}^\top S(\tilde\theta_k) x_{t_{k+1}}\Big]\Big]
\notag\\
\leq &\ee\Big[\sum_{k=1}^{K_T}
x^\top_{t_k} S(\tilde\theta_k) x_{t_k}\Big].
\end{align}
Since $||S(\tilde\theta_{k})|| \leq M_S$, we obtain
\begin{align}
R_1 \leq &\ee\Big[\sum_{k=1}^{K_T} M_S ||x_{t_k}||^2 \Big]
\leq M_S\ee\Big[K_T X_T^2 \Big].
\end{align}
\end{proof}

We now derive an upper bound for $R_2$.
\begin{lemma}
\label{lm:R2}
The third term $R_2$ is bounded by
\begin{align}
R_2 \leq 
& O\Big( 
M_2
\sqrt{(T + \ee[K_T])\ee[X_T^4\log(T X_T^2)]}
\Big)
\end{align}
where $M_2 = M_SM_\theta M_G^2
\sqrt{\frac{32d^2n}{\lambda_{min}}}$ and $\lambda_{min}$ is the minimum eigenvalue of $\Sigma^{-1}_1$.
\end{lemma}
\begin{proof}
Each term inside the expectation of $R_2$ is equal to
\begin{align}
& ||S^{0.5}(\tilde\theta_k) \theta_1^\top z_t ||^2-||S^{0.5}(\tilde\theta_k) \tilde\theta_k^\top z_t ||^2
\notag\\
=& \Big(||S^{0.5}(\tilde\theta_k)\theta_1^\top z_t ||+||S^{0.5}(\tilde\theta_k) \tilde\theta_k^\top z_t ||\Big)
\notag\\
&\quad \Big(||S^{0.5}(\tilde\theta_k) \theta_1^\top z_t ||-||S^{0.5}(\tilde\theta_k)\tilde\theta_k^\top z_t ||\Big)
\notag\\
\leq& \Big(||S^{0.5}(\tilde\theta_k) \theta_1^\top z_t ||+||S^{0.5}(\tilde\theta_k) \tilde\theta_k^\top z_t ||\Big)
\notag\\
&||S^{0.5}(\tilde\theta_k) (\theta_1-\tilde\theta_k)^\top z_t ||.
\end{align}
Since $||S^{0.5}(\tilde\theta_k) \theta^\top z_t|| \leq M_S^{0.5} M_\theta M_G X_T$ for $\theta = \tilde\theta_k$ or $\theta = \theta_1$, the above term can be further bounded by
\begin{align}
& 2M_S^{0.5}M_\theta M_G X_T ||S^{0.5}(\tilde\theta_k)(\theta_1-\tilde\theta_k )^\top z_t ||
\notag\\
\leq & 2M_S M_\theta M_G X_T ||(\theta_1-\tilde\theta_k )^\top z_t ||.
\end{align}
Therefore,
\begin{align}
R_2 \leq & 2M_SM_\theta M_G\ee\Big[ X_T\sum_{k=1}^{K_T}\sum_{t=t_k}^{t_{k+1}-1}  ||(\theta_1-\tilde\theta_k )^\top z_t || \Big].
\label{eq:boundR2_1}
\end{align}
From Cauchy-Schwarz inequality, we have
\begin{align}
&\ee\Big[X_T\sum_{k=1}^{K_T}\sum_{t=t_k}^{t_{k+1}-1}  ||(\theta_1-\tilde\theta_k )^\top z_t || \Big]
\notag\\
=
&\ee\Big[X_T\sum_{k=1}^{K_T}\sum_{t=t_k}^{t_{k+1}-1}  ||(\Sigma^{-0.5}_t(\theta_1-\tilde\theta_k ))^\top  \Sigma^{0.5}_tz_t || \Big]
\notag\\
\leq
&\ee\Big[\sum_{k=1}^{K_T}\sum_{t=t_k}^{t_{k+1}-1}  ||\Sigma^{-0.5}_t(\theta_1-\tilde\theta_k )|| \times  X_T||\Sigma^{0.5}_tz_t || \Big]
\notag\\
\leq 
&\sqrt{\ee\Big[\sum_{k=1}^{K_T}\sum_{t=t_k}^{t_{k+1}-1}  ||\Sigma^{-0.5}_t(\theta_1-\tilde\theta_k )||^2\Big]}
\notag\\
&\sqrt{\ee\Big[\sum_{k=1}^{K_T}\sum_{t=t_k}^{t_{k+1}-1}  X_T^2||\Sigma^{0.5}_t z_t||^2\Big]}
\label{eq:boundR2_2}
\end{align}

From Lemma \ref{lm:ineq_sigmatheta} in the appendix, the first part of \eqref{eq:boundR2_2} is bounded by
\begin{align}
&\sqrt{\ee\Big[\sum_{k=1}^{K_T}\sum_{t=t_k}^{t_{k+1}-1}  ||\Sigma^{-0.5}_t(\theta_1-\tilde\theta_k )||^2\Big]}
\notag\\
\leq & \sqrt{4dn (T + \ee[K_T])}.
\label{eq:boundR2_part1}
\end{align}

For the second part of \eqref{eq:boundR2_2}, note that
\begin{align}
 \sum_{k=1}^{K_T}\sum_{t=t_k}^{t_{k+1}-1}  ||\Sigma^{0.5}_t z_t||^2
=  \sum_{t=1}^T  z_t^\top \Sigma_t z_t
\end{align}
Since $||z_t|| \leq M_G X_T$ for all $t \leq T$, Lemma 8 of \cite{abbasi2015bayesian} implies
\begin{align}
&\sum_{t=1}^T  z_t^\top \Sigma_t z_t
\notag\\
\leq &\sum_{t=1}^T \max(1,M_G^2 X_T^2/\lambda_{min}) \min(1,z_t^\top \Sigma_t z_t)
\notag\\
\leq & 2d \max(1,M_G^2 X_T^2/\lambda_{min})
\log(\tr(\Sigma^{-1}_1)+T M_G^2 X_T^2)
\notag\\
= & O\Big( \frac{2d M_G^2}{\lambda_{min}} X_T^2\log(T X_T^2) \Big).
\end{align}
Consequently, the second term of \eqref{eq:boundR2_2} is bounded by
\begin{align}
&O\Big( \sqrt{ \frac{2d M_G^2 }{\lambda_{min}}\ee\Big[ X_T^4\log(T X_T^2) \Big]} \Big).
\label{eq:boundR2_part2}
\end{align}
Then, from \eqref{eq:boundR2_1}, \eqref{eq:boundR2_2}, \eqref{eq:boundR2_part1} and \eqref{eq:boundR2_part2}, we obtain the result of the lemma.
\end{proof}

Using the bounds on $R_0$, $R_1$ and $R_2$, we are now ready to prove Theorem \ref{thm:regret}.
\begin{proof}[Proof of Theorem \ref{thm:regret}]
From the regret decomposition \eqref{eq:regret_decomp}, Lemmas \ref{lm:R0}-\ref{lm:R2}, and the bound on $K_T$ from Lemma \ref{lm:KT}, we obtain
\begin{align}
&R(T,\text{{\tt TSDE}})
\notag\\
\leq &O\Big(
M_2 \sqrt{(T + \ee[ \sqrt{ 2d T\log(T X_T^2)} ])\ee[X_T^4\log(T X_T^2)]}
\notag\\
& +  \ee\Big[  \sqrt{ 2d T\log(T X_T^2)} (M_J+ M_S X_T^2) \Big]
\Big)
\notag\\
\leq &\tilde O\Big(
\sqrt{(T + \ee[ \sqrt{  T\log(X_T)} ])\ee[X_T^4\log(X_T)]}
\notag\\
& +  \ee\Big[  \sqrt{T\log(X_T)} X_T^2 \Big]
\Big).
\label{eq:regretboundinproof}
\end{align}
From Lemma \ref{lm:bounds_some_expectations} in the appendix, we have
$\ee[\sqrt{\log(X_T)}]\leq  \tilde O(1) $, 
$\ee[\sqrt{\log(X_T)}X_T^2]
\leq \tilde O\Big( (1-\delta)^{-2} \Big)$, and $\ee[X_T^4\log(X_T)]
\leq  \tilde O\Big((1-\delta)^{-5}  \Big)$.
Applying these bounds to \eqref{eq:regretboundinproof} we get
\begin{align}
&R(T,\text{{\tt TSDE}})
\notag\\
\leq & \tilde O\Big(
\sqrt{(T +  \sqrt{  T} ) (1-\rho)^{-5}}
 +   \sqrt{T}(1-\delta)^{-2}
\Big)
\notag\\
= & \tilde O\Big(
\sqrt{T}(1-\delta)^{-2.5}
\Big).
\end{align}
\end{proof}

\section{Regret Analysis for Time-Varying Parameter}\label{sec:timevarying}

We now present the regret analysis for the time-varying parameter case.

Since the true parameter varies over time, we make a stronger assumption on the prior distribution to ensure stability.
\begin{assumption}
\label{assum:stabilizability_TV}
There exists a positive number $\delta <1$ such that for any $\theta, \theta'\in \Omega_1$, we have
$\rho(A+ BG(\theta')) \leq \delta$ where $\theta^\top = [A,B]$.
\end{assumption}

From the assumption, the closed-loop system is stable for any model parameter under the learning algorithm.

Let $N_T = \ee\Big[\sum_{t\leq T} j_t\Big]$ be the expected number of jumps up to time $T$.
We now present the bound on expected regret of {\tt TSDE} in the time-varying parameter case.
\begin{theorem}
\label{thm:regret_TV}
Under Assumptions \ref{assum:prior} and \ref{assum:stabilizability_TV}, the expected regret \eqref{eq:regret_TV} of {\tt TSDE-TV} satisfies
\begin{align}
R_{TV}(T,\text{{\tt TSDE-TV}}) \leq \tilde O\Big(
T^{\frac{2+q}{2+2q}} + T^{\frac{q}{1+q}}N_T
\Big)
\end{align}
where $\tilde O(\cdot)$ hides all constants and logarithmic factors.
\end{theorem}

From Theorem \ref{thm:regret_TV}, we have the following corollary.
\begin{corollary}\label{corrolary:TV parameter}
If $N_T \leq T^\alpha$ for some $\alpha <1$, we can pick $q = \frac{2(1-\alpha)}{1+2\alpha}$ such that
\begin{align}
R_{TV}(T,\text{{\tt TSDE-TV}}) \leq \tilde O\Big( T^{\frac{2+\alpha}{3}}
\Big).
\end{align}
\end{corollary}

The corollary says that when the expected number of jumps is sub-linear in $T$, with an appropriate choice of algorithm parameter $q$, the {\tt TSDE-TV} algorithm  can achieve a sub-linear growth of expected regret.

\begin{remark}
Note that the sub-linear regret growth of {\tt TSDE-TV} implies its asymptotically optimal performance under the average cost criterion.
\end{remark}

We proceed to analyze the regret of {\tt TSDE-TV} and prove Theorem \ref{thm:regret_TV}.

Let $L_T$ be the number of re-initializations of {\tt TSDE-TV} upto time $T$.
Then, we can divide the time horizon $T$ into $L_T+1$ phases. Let $s_1, s_2,\dots,s_{L_T}$ be the times {\tt TSDE-TV} re-initializes. Then the $l$th phase has length $D_l = s_l - s_{l-1}$ for $l=1,2,\dots,(L_T+1)$ (with the convention $s_0 = 1$ and $s_{L_T+1} = T+1$).  
From the specification of {\tt TSDE-TV} we have $D_l = l^q$ for $l \leq L_T$.

We can now decompose the regret of {\tt TSDE-TV} into two parts. The first part is the performance loss during the phases without parameter jumps, and the second part is the loss during the phases with at least one parameter jump. Specifically, we have
\begin{align}
R_{TV}(T,\text{{\tt TSDE-TV}}) = R_{TV,0} + R_{TV,1}
\end{align}
where
\begin{align}
& R_{TV,0} = \ee\Big[ \sum_{l\leq L_T+1: j_t = 0 \text{ for all } s_{l-1}\leq t < s_l}  \sum_{t = s_{l-1}}^{s_{l}-1}  \Big[ c_t  - J(\theta_t) \Big] \Big]
\\
& R_{TV,1} = \ee\Big[ \sum_{l\leq L_T+1: j_t = 1 \text{ for some } s_{l-1}\leq t < s_l}  \sum_{t = s_{l-1}}^{s_{l}-1}  \Big[ c_t  - J(\theta_t) \Big] \Big]
\end{align}

We  use the stationary parameter result to bound the first part $R_{TV,0}$.
\begin{lemma}
\label{lm:RTV0}
\begin{align}
R_{TV,0} \leq \tilde O( L_T^{\frac{2+q}{2}}).
\end{align}
\end{lemma}
\begin{proof}
Since the model parameter remains the same for each phase $l$ in $R_{TV,0}$, the system during each of such phase is the same as the stationary parameter case. Therefore, the regret analysis for TSDE can be applied here because {\tt TSDE-TV} is the same as TSDE during a phase. Note that the length of phase $l$ is $D_l$, so from Theorem \ref{thm:regret} we get
\begin{align}
R_{TV,0} = &\ee\Big[ \sum_{l\leq L_T+1: j_t = 0 \text{ for all } s_{l-1}\leq t < s_l}  
\notag\\
& \ee\Big[\sum_{t = s_{l-1}}^{s_{l}-1}  \Big[ c_t  - J(\theta_t) \Big]|j_t = 0 \text{ for all } s_{l-1}\leq t < s_l\Big] \Big]
\notag\\
\leq &\ee\Big[ \sum_{l\leq L_T+1: j_t = 0 \text{ for all } s_{l-1}\leq t < s_l}  
\tilde O(\sqrt{D_l}) \Big]
\notag\\
\leq & \tilde O( \sum_{l\leq L_T+1} \sqrt{D_l})
\notag\\
\leq & \tilde O( \sum_{l\leq L_T+1} l^{\frac{q}{2}})
= \tilde O( L_T^{\frac{2+q}{2}}).
\label{eq:pfinRTV0}
\end{align}
\end{proof}

Using the expected number of parameter jumps $N_T$ we bound the second part of the regret in the lemma below.
\begin{lemma}
\label{lm:RTV1}
\begin{align}
R_{TV,1} \leq \tilde O(L_T^q N_T).
\end{align}
\end{lemma}
\begin{proof}
Note that for any $t\leq T$ we have
\begin{align}
& c_t - J(\theta_t) \leq c_t = x_t^\top Q x_t + u_t^\top R u_t
\notag\\
\leq & ||Q|| \cdot ||x_t||^2  + \max_{\theta \in \Omega_1}||R|| \cdot ||G(\theta)|| \cdot ||x_t||^2
\notag\\
\leq & \tilde O ( X_T^2).
\end{align}
Therefore, 
\begin{align}
 R_{TV,1} \leq &\ee\Big[ \sum_{l\leq L_T+1: j_t = 1 \text{ for some } s_{l-1}\leq t < s_l}  D_l \tilde O ( X_T^2) \Big]
 \notag\\
 \leq & (L_T+1)^q\ee\Big[ \sum_{t \leq T}  j_t  \tilde O ( X_T^2) \Big]
 \label{eq:pfinRTV1_1}
\end{align}
where the last inequality holds because $D_l \leq (L_T+1)^q$.

Since the jump process is independent of the system noise, the proof of Lemma \ref{lm:bound_X} also hold when conditioned on $\sum_{t \leq T}  j_t$. That is,
\begin{align}
\ee\Big[ X_T^2| \sum_{t \leq T}  j_t\Big] \leq \tilde O(1).
 \label{eq:pfinRTV1_2}
\end{align}
From \eqref{eq:pfinRTV1_1} and \eqref{eq:pfinRTV1_1} we get
\begin{align}
 R_{TV,1} \leq &(L_T+1)^q \ee\Big[ \sum_{t \leq T}  j_t \tilde O ( X_T^2) \Big]
 \notag\\
 = &(L_T+1)^q \ee\Big[ \sum_{t \leq T}j_t \ee\Big[ \tilde O ( X_T^2)|\sum_{t \leq T}j_t\Big] \Big]
 \notag\\
\leq &(L_T+1)^q \ee\Big[ \sum_{t \leq T}j_t \Big]
= \tilde O(L_T^q N_T).
\end{align}

\end{proof}

We now prove Theorem \ref{thm:regret_TV} using the above lemmas.
\begin{proof}[Proof of Theorem \ref{thm:regret_TV}]
From Lemmas \ref{lm:RTV0} and \ref{lm:RTV1} we have
\begin{align}
R_{TV}(T,\text{{\tt TSDE-TV}}) \leq \tilde O\Big(
L_T^{\frac{2+q}{2}}+ L_T^q N_T
\Big).
\label{eq:inpfthmTV}
\end{align}

Note that 
\begin{align}
T \geq \sum_{l \leq L_T} D_l = \sum_{l \leq L_T} l^q = \tilde O(L_T^{1+q}).
\end{align}
So $L_T \leq \tilde O(T^{\frac{1}{1+q}})$, and the proof of theorem is complete by applying this bound on $L_T$ in \eqref{eq:inpfthmTV}.
\end{proof}

\section{Simulations}
\subsection{Stationary Parameter}

In this section, we illustrate through numerical simulations the performance of the {\tt TSDE} algorithm for different linear systems.
The prior distribution used in {\tt TSDE} are set according to \eqref{eq:mu1} with $\hat \theta_1(i) = 1$, $\Sigma_1 = I$, and $\Omega_1 = \{\theta: \rho(A_1+ B_1 G(\theta)) \leq \delta \}$ where $\delta$ is a simulation parameter. 
The parameter $\delta$ can be seen as the level of accuracy of the prior distribution. The smaller $\delta$ is, the more accurate the prior distribution is for the true system parameters.
Note that Assumption \ref{assum:stabilizability} holds when $\delta< 1$, but it does not hold when $\delta \geq 1$.

For each system, we select $\delta =0.99$ and $\delta=2$. We run $500$ simulations and show the mean of regret with confidence interval for each scenario.

In the case of a scalar system, we consider two systems: a stable system with $A_1 = 0.9$ and an unstable system with $A_1 = 1.5$. 
We set $Q=2$, $R=1$ and $ B_1 = 0.5$ for both cases. Figure \ref{fig_scalar}(\subref{fig_stable_scalar}) shows the results for the stable system and figure \ref{fig_scalar}(\subref{fig_unstable_scalar}) for the unstable system. 
{\tt TSDE} successfully learns and controls both the stable and the unstable system as the regret grows at a sublinear rate (though not apparent, it grows as $\tilde{O}(\sqrt{T})$).
Although Assumptions \ref{assum:stabilizability} does not hold when $\delta=2$, the results show that {\tt TSDE} might still work in this situation.

\begin{figure}[t]
	\centering        		
	\begin{subfigure}[b]{0.45\textwidth}
		\centering
		\includegraphics[width=\textwidth]{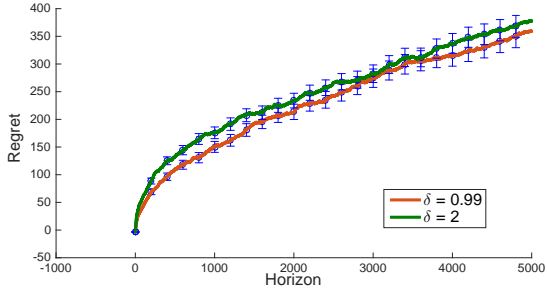}
		\caption{\small{Expected regret for a stable scalar system}}
		\label{fig_stable_scalar}
	\end{subfigure} 
	\begin{subfigure}[b]{0.45\textwidth}
		\centering
		\includegraphics[width=\textwidth]{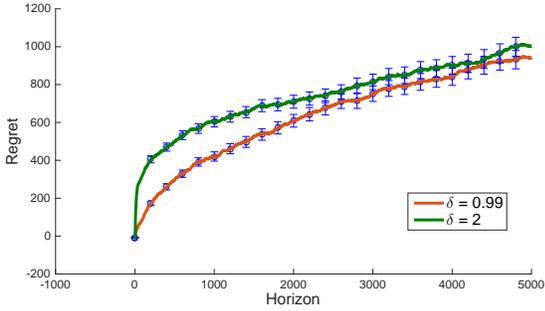}
		\caption{\small{Expected Regret for a unstable scalar system}}
		\label{fig_unstable_scalar}
	\end{subfigure}  
	\caption{\small{Scaler Systems}}	
	\label{fig_scalar}
\end{figure}

Figure \ref{fig_vector} illustrates the regret curves for a multi-dimensional system with $n=m=3$. We again consider two systems: a stable system with $0.9$ as the largest eignevalue of $A_1$ and an unstable system with $1.5$ as the largest eigenvalue of $A_1$.  The results show that {\tt TSDE} achieves sublinear regret in the multi-dimensional cases also. In fact, it can be verified that the rate of growth matches with the theoretical rate of theorem \ref{thm:regret}, $\tilde{O}(\sqrt{T})$. 

\begin{figure}[h]
\centering        		
    \begin{subfigure}[t]{0.48\textwidth}
        \centering
        \includegraphics[width=\textwidth]{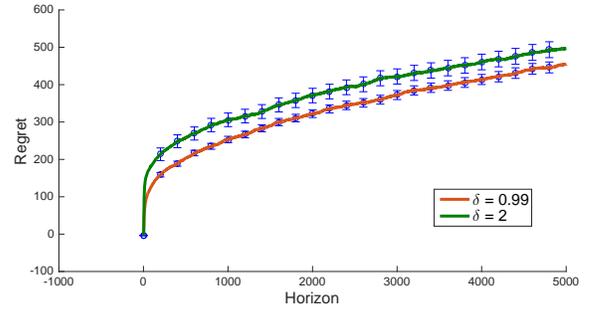}
        \caption{\small{Expected regret for a stable vector system with $n=m=3$}}
        \label{fig_stable_vector}
    \end{subfigure}    	
    \begin{subfigure}[t]{0.48\textwidth}
        \centering
        \includegraphics[width=\textwidth]{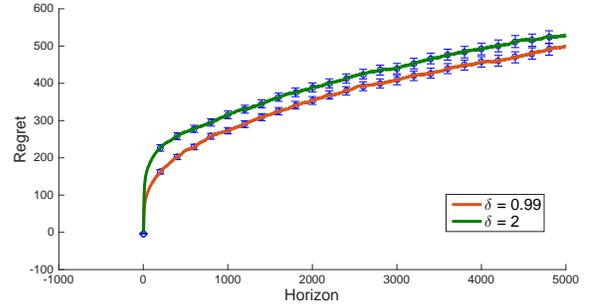}
        \caption{\small{Expected Regret for a unstable vector system with $n=m=3$}}
        \label{fig_unstable_vector}
    \end{subfigure}  
\caption{\small{Multi-Dimensional Systems}}	
\label{fig_vector}
\end{figure}

\subsection{Time-Varying Parameter}

Now, we present simulation results for the case when the true parameter is non stationary. We set the prior distribution of the true paramter to $N(\hat{\theta}_1, \Sigma_1)$ projected on the set $\Omega_1 = \{\theta : ||\theta - \theta_{prior}|| < \epsilon \}$ where $\theta_{prior} = [A_{prior}, B_{prior}]$ can be interpreted as the prior belief around which $\theta_t$ takes values. When $\epsilon$ is small, then all the $\theta_t$'s are close to each other. Assumption  \ref{assum:stabilizability_TV} is satisfied when $\epsilon$ is small due to the continuity of the spectral radius. 

We run the simulations for a horizon of $T = 50000$ and set $\alpha = 0.2$.  $\theta_t$ is generated randomly from its prior distribution at each change point. The number of jumps in $\theta_t$ are fixed to $\lfloor T^{\alpha} \rfloor$ and the change points are distributed uniformly across the horizon for the purpose of simulation. We plot the mean of the regret with its confidence interval over $200$ simulation runs for  $\epsilon = 0.5$ and $\epsilon = 0.8$ in figure \ref{fig_parameter_change_scalar}.

 We set $\hat{\theta}_1 = \theta_{prior}$, $\Sigma_1 = \frac{1}{100} I$, and $q$ as in Corollary \ref{corrolary:TV parameter}. Figure \ref{fig_parameter_change_scalar}(\subref{regret}) shows the regret curve for a scalar system with  $A_{prior} = 1$ and  $B_{prior} = 0.5$. It can be observed that the higher value of $\epsilon$ results in higher regret. This is because smaller $\Omega_1$ would mean less variation among the $\theta_t$ when it changes. Smaller jumps in $\theta_t$ would imply less accumulated regret. 
 
Figure \ref{fig_parameter_change_scalar}(\subref{per_unit_regret}) shows the behaviour of regret per unit time $\frac{R(T)}{T}$. The curve can be seen decreasing to $0$ for both $\epsilon = 0.5, 0.8$ which shows that {\tt TSDE-TV} achieves sub-linearly growing cumulative regret. 
Also, it can be verified that the rate of growth is consistent with the theoretical limit as given in corollary \ref{corrolary:TV parameter}

Figure \ref{fig_parameter_change_vector} shows the analogous results for multi-dimensional system with $n=m=3$. In this case the eigen values of $A_{prior}$ are set to $1,0.7$ and $-0.2$. We can observe a similar behaviour in the regret curves and regret per unit time curves as in the scalar case. This establishes that {\tt TSDE-TV} achieves sub-linear cumulative regret in the multi-dimensional case also. 

 \begin{figure}[t]
 	\centering        		
 	\begin{subfigure}[t]{0.48\textwidth}
 		\centering
 		\includegraphics[width=\textwidth]{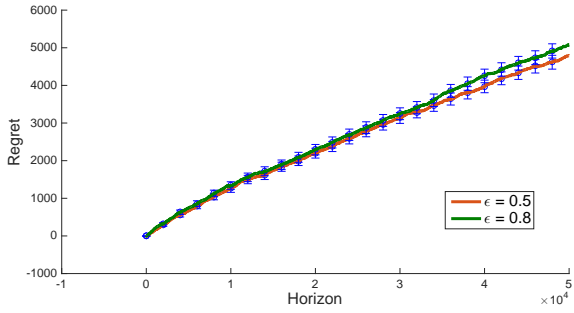}
 		\caption{\small{Regret for a TV scalar system}}
 		\label{regret}
 	\end{subfigure}    	
 	\begin{subfigure}[t]{0.48\textwidth}
 		\centering
 		\includegraphics[width=\textwidth]{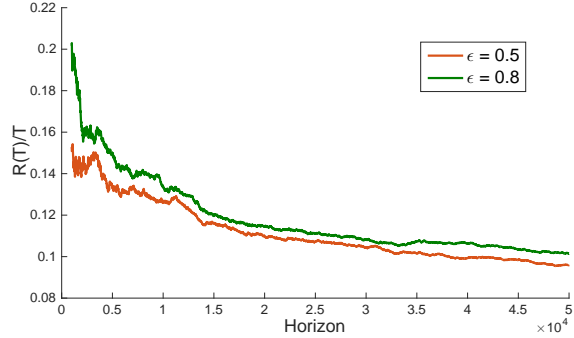}
 		\caption{\small{Regret per unit time for a TV scalar system}}
 		\label{per_unit_regret}
 	\end{subfigure}  
 	\caption{\small{Scalar TV system}}	
 	\label{fig_parameter_change_scalar}
 \end{figure}
 
 \begin{figure}[t]
 	\centering        		
 	\begin{subfigure}[t]{0.48\textwidth}
 		\centering
 		\includegraphics[width=\textwidth]{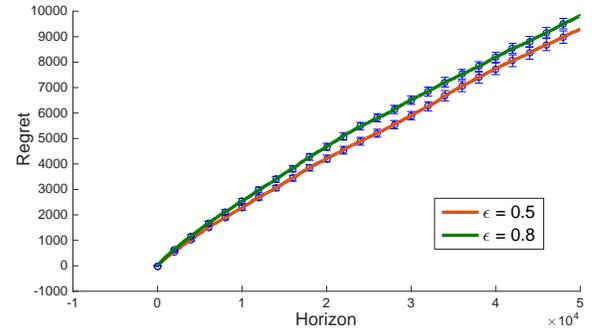}
 		\caption{\small{Regret for a TV vector system}}
 		\label{regret}
 	\end{subfigure}    	
 	\begin{subfigure}[t]{0.48\textwidth}
 		\centering
 		\includegraphics[width=\textwidth]{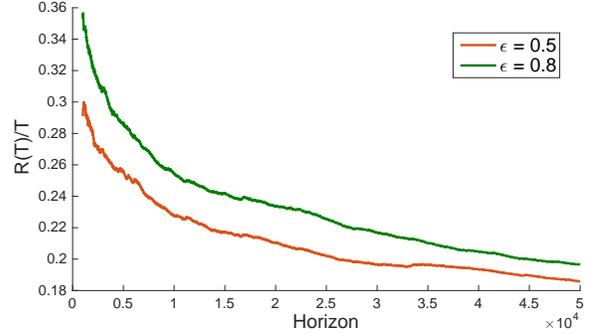}
 		\caption{\small{Regret per unit time for a TV vector system}}
 		\label{per_unit_regret}
 	\end{subfigure}  
 	\caption{\small{Multi-dimensional TV system}}	
 	\label{fig_parameter_change_vector}
 \end{figure}


\section{Conclusion}
\label{sec:conclusion}

In this paper, we have proposed a Thompson sampling with dynamic episodes (TSDE) learning algorithm for control of stochastic linear systems with quadratic costs. Under some conditions on the prior distribution, we provide a  $\tilde O(\sqrt{T})$ bound on expected regret of TSDE. This implies that the average regret per unit time goes to zero, thus the learning algorithm asymptotically learns the optimal control policy. We believe this is the first near-optimal guarantee on expected regret for a learning algorithm in LQ control. We have also shown that TSDE with a re-initialization schedule (i.e. the TSDE-TV algorithm) is robust to time-varying drift in model parameters. As long as the drift is not large, the algorithm can ``track" the model parameters and find an approximately optimal control law. Numerical simulations confirm that TSDE indeed achieves sublinear regret which matches with the theoretical upper bounds. In addition to use of the Thompson sampling-based learning, the key novelty here is design of an exploration schedule  that achieves sublinear regret.


\bibliographystyle{ieeetr}
\bibliography{IEEEabrv,ref_learning}


\appendix
\begin{proof}[Proof of Lemma \ref{lm:bound_X}]
During the $k$th episode, we have $u_t = G(\tilde{\theta}_k)x_t$. Then,
\begin{align}
||x_{t+1}|| = & ||(A_1+B_1G(\tilde{\theta}_k))x_t + w_t||
\notag\\
\leq &||(A_1+B_1G(\tilde{\theta}_k))x_t|| + ||w_t||
\notag\\
\leq &\rho(A_1+B_1G(\tilde{\theta}_k))||x_t|| + ||w_t||
\notag\\
\leq & \delta ||x_t|| + ||w_t||
\label{eq:boundxnext}
\end{align}
where the second inequality is the property of spectral radius, and the last inequality follows from Assumption \ref{assum:stabilizability}.
Iteratively applying \eqref{eq:boundxnext}, we get
\begin{align}
\lVert x_{t} \Vert 
\leq &\sum_{\tau < t} \delta^{t-\tau-1} \lVert w_\tau\rVert
\notag
\leq & \sum_{\tau < t} \delta^{t-\tau-1} \max_{\tau \leq T}\lVert w_\tau\rVert
\notag\\
\leq & \frac{1}{1-\delta} \max_{\tau \leq T}\lVert w_\tau\rVert.
\end{align}
Therefore,
\begin{align}
X_T^j
\leq &\Big(\frac{1}{1-\delta} \max_{t \leq T}\lVert w_t\rVert \Big)^j
\notag\\
= & (1-\delta)^{-j} \max_{t \leq T}\lVert w_t\rVert^j.
\label{eq:boundxfromw}
\end{align}
Then, it remains to bound $\ee[\max_{t \leq T}\lVert w_t\rVert^j]$. Following the steps of \cite{Gautam_BoundforMax}, we have
\begin{align}
 \exp\Big(\ee[\max_{t \leq T}\lVert w_t\rVert^j]\Big)
\leq & 
\ee\Big[\exp\Big(\max_{t \leq T}\lVert w_t\rVert^j\Big)\Big]
\notag\\
= & \ee\Big[\max_{t \leq T}\exp\Big(\lVert w_t\rVert^j\Big)\Big]
\notag\\
\leq & \ee\Big[\sum_{t \leq T}\exp\Big(\lVert w_t\rVert^j\Big)\Big]
\notag\\
= & T \ee\Big[\exp\Big(\lVert w_1\rVert^j\Big)\Big].
\label{eq:boundforw}
\end{align}
Combining \eqref{eq:boundxfromw} and \eqref{eq:boundforw}, we obtain
\begin{align}
\ee[X_T^j]
\leq &(1-\delta)^{-j}\log\Big(T \ee\Big[\exp\Big(\lVert w_1\rVert^j\Big)\Big]\Big)
\notag\\
= &O\Big((1-\delta)^{-j}\log(T)\Big).
\end{align}

\end{proof}

\begin{proof}[Proof of Lemma \ref{lm:PS_expectation}]
From the definition \eqref{eq:tk}, the start time $t_k$ is a stopping-time, i.e. $t_k$ is $\sigma(h_{t_k})-$measurable.
Note that $\tilde{\theta}_{k}$ is randomly sampled from the posterior distribution $\mu_{t_k}$.
Since ${t_k}$ is a stopping time, ${t_k}$ and $\mu_{t_k}$ are both measurable with respect to $\sigma(h_{t_k})$.
From the condition, $X$ is also measurable with respect to $\sigma(h_{t_k})$.
Then, conditioned on $h_{t_k}$, the only randomness in $f(\tilde{\theta}_{k},X)$ is the random sampling in the algorithm. This gives the following equation:
\begin{align}
\ee\Big[ f(\tilde{\theta}_{k},X)|h_{t_k} \Big]
= &\ee\Big[ f(\tilde{\theta}_{k},X)|h_{t_k},{t_k},\mu_{t_k} \Big] 
\notag\\
= & \int f(\theta,X)\mu_{t_k}(d \theta)
=\ee\Big[ f(\theta_1,X)|h_{t_k} \Big]
\label{eq:lm1L}
\end{align}
since $\mu_{t_k}$ is the posterior distribution of $\theta_1$ given $h_{t_k}$.
Now the result follows by taking the expectation of both sides.
\end{proof}

\begin{proof}[Proof of Lemma \ref{lm:KT}]
Define macro-episodes with start times $t_{n_i}, i=1,2,\dots$ where $t_{n_1} = t_1$ and 
\begin{align*}
t_{n_{i+1}} = \min\{ & t_k > t_{n_i}:\quad 
\det(\Sigma_{t_k}) < 0.5\det(\Sigma_{t_{k-1}}) \}.
\end{align*}
The idea is that each macro-episode starts when the second stopping criterion happens.
Let $M$ be the number of macro-episodes until time $T$
and define $n_{(M+1)} = K_T+1$.
Let $\mathcal M$ be the set of episodes that is the first one in a macro-episode.

Let $\tilde T_i = \sum_{k=n_{i}}^{n_{i+1}-1} T_{k}$ be the length of the $i$th macro-episode.
By  definition of macro-episodes, any episode except the last one in a macro-episode must be triggered by the first stopping criterion. Therefore, within the $i$th macro-episode, $T_{k} = T_{k-1} +1$ for all $k = n_{i}, n_{i}+1,\dots,n_{i+1}-2$.
Hence,
\begin{align*}
\tilde T_i = &\sum_{k=n_{i}}^{n_{i+1}-1} T_{k}
= \sum_{j=1}^{n_{i+1}-n_i-1} (T_{n_i-1}+j) + T_{n_{i+1}-1}
\notag\\ 
\geq &\sum_{j=1}^{n_{i+1}-n_i-1} (j+1) + 1 = 0.5(n_{i+1}-n_i)(n_{i+1}-n_i+1).
\end{align*}
Consequently, $n_{i+1} - n_{i} \leq \sqrt{2 \tilde T_i}$ for all $i=1,\dots,M$.
From this property, we obtain
\begin{align}
K_T = &n_{M+1}-1
 = \sum_{i=1}^{M} (n_{i+1} - n_{i})
 \leq \sum_{i=1}^{M} \sqrt{2 \tilde T_i} .
 \label{eq:KTboundinpf}
 \end{align}
Using \eqref{eq:KTboundinpf} and the fact that $\sum_{i=1}^M \tilde T_i = T$, we get
\begin{align}
K_T \leq \sum_{i=1}^{M} \sqrt{2 \tilde T_i} \leq & 
\sqrt{ M\sum_{i=1}^{M} 2 \tilde T_i } 
= \sqrt{ 2 MT}
\label{eq:KTboundfromM}
\end{align}
where the second inequality is by Cauchy-Schwarz. 

Since the second stopping criterion is triggered whenever the determinant of sample covariance is half, we have
\begin{align}
&\det(\Sigma^{-1}_{T}) \geq \det(\Sigma^{-1}_{t_{n_M}}) > 2\det(\Sigma^{-1}_{t_{N_M-1}}) 
\notag\\
&>\dots> 2^{M-1}\det(\Sigma^{-1}_{1})
\end{align}
Since $(\tr(\Sigma_T^{-1}))^d \geq \det(\Sigma_T^{-1})$, we have
\begin{align}
&\tr(\Sigma_T^{-1}) > (\det(\Sigma_T^{-1}))^{1/d} 
\notag\\
> &2^{(M-1)/d} (\det(\Sigma^{-1}_{1}))^{1/d}
\geq 2^{(M-1)/d} \lambda_{min}
\end{align}
where $\lambda_{min}$ is the minimum eigenvalue of $\Sigma^{-1}_1$.
Note that from Remark \ref{rm:invcov},
\begin{align}
\Sigma_T^{-1} = \Sigma_1^{-1} + \sum_{t=1}^{T-1} z_t z_t^\top
\end{align}
and we obtain
\begin{align}
2^{(M-1)/d} \lambda_{min} <\tr(\Sigma_1^{-1}) + \sum_{t=1}^{T-1} z_t^\top z_t.
\end{align}
Then,  
\begin{align}
M \leq & 1+ d\log(\frac{1}{\lambda_{min}}(\tr(\Sigma_1^{-1}) + \sum_{t=1}^{T-1} z_t^\top z_t))
\notag\\
= &O\left( d\log(\sum_{t=1}^{T-1} z_t^\top z_t)\right).
\end{align}
Note that, $||z_t|| = ||[I, G(\theta)^\top]^\top x_t|| \leq M_G ||x_t||$.
Consequently,
\begin{align}
M \leq& O\Big(d \log\Big(M_G^2 \sum_{t=1}^{T-1}||x_t||^2\Big) \Big)
=O\Big(d \log\Big( \sum_{t=1}^{T-1}||x_t||^2\Big) \Big)
\notag\\
\leq &O\Big(d \log(TX_T^2) \Big)
\end{align}
Hence, from \eqref{eq:KTboundfromM} we obtain the claim of the lemma.
\end{proof}

\begin{lemma}
\label{lm:ineq_sigmatheta}
We have the following inequality:
\begin{align}
\ee \Big[\sum_{k=1}^{K_T}\sum_{t=t_k}^{t_{k+1}-1}  ||\Sigma^{-0.5}_t(\theta_1-\tilde{\theta}_k )||^2\Big]
\leq &  4dn (T + \ee[K_T]).
\end{align}
\end{lemma}
\begin{proof}
From Lemma 9 of \cite{abbasi2015bayesian}, we have
\begin{align}
&||\Sigma^{-0.5}_t (\theta_1 -\tilde{\theta}_k)||^2
\notag\\
\leq &
||\Sigma^{-0.5}_{t_k} (\theta_1 -\tilde{\theta}_k)||^2\frac{\det(\Sigma_{t_k})}{\det(\Sigma_{t})}
\notag\\
\leq & 2
||\Sigma^{-0.5}_{t_k} (\theta_1 -\tilde{\theta}_k )||^2
\end{align}
where the last inequality follows from the second stopping criterion of the algorithm.
Therefore, 
\begin{align}
&\sum_{k=1}^{K_T}\sum_{t=t_k}^{t_{k+1}-1}  ||\Sigma^{-0.5}_t(\theta_1-\tilde{\theta}_k )||^2
\notag\\
\leq 
&2 \sum_{k=1}^{K_T} T_k ||\Sigma^{-0.5}_{t_k} (\theta_1 -\tilde{\theta}_k)||^2.
\end{align}
Using the idea of the proof of Lemma \ref{lm:R0}, we obtain
\begin{align}
& \ee \Big[\sum_{k=1}^{K_T} T_k ||\Sigma^{-0.5}_{t_k} (\theta_1-\tilde{\theta}_k)||^2\Big]
\notag\\
= & \sum_{k=1}^{\infty}\ee \Big[  \mathds{1}_{\{t_{k}\leq T\}}T_k ||\Sigma^{-0.5}_{t_k} (\theta_1 -\tilde{\theta}_k)||^2\Big]
\notag\\
\leq &\sum_{k=1}^{\infty}\ee \Big[\mathds{1}_{\{t_{k}\leq T\}} (T_{k-1}+1) ||\Sigma^{-0.5}_{t_k} (\theta_1 -\tilde{\theta}_k)||^2\Big].
\end{align}
Since $\mathds{1}_{\{t_{k}\leq T\}} (T_{k-1}+1)$ is measurable with respect to $\sigma(h_{t_k})$, we get
\begin{align}
& \ee \Big[\mathds{1}_{\{t_{k}\leq T\}} (T_{k-1}+1) ||\Sigma^{-0.5}_{t_k} (\theta_1-\tilde{\theta}_k)||^2\Big]
\notag\\
=&\ee \Big[ \ee \Big[\mathds{1}_{\{t_{k}\leq T\}} (T_{k-1}+1) ||\Sigma^{-0.5}_{t_k} (\theta_1 -\tilde{\theta}_k)||^2|h_{t_k}\Big]\Big]
\notag\\
=&\ee \Big[\mathds{1}_{\{t_{k}\leq T\}} (T_{k-1}+1)  \ee \Big[||\Sigma^{-0.5}_{t_k} (\theta_1 -\tilde{\theta}_k)||^2|h_{t_k}\Big]\Big]
\notag\\
\leq &\ee \Big[\mathds{1}_{\{t_{k}\leq T\}} (T_{k-1}+1) 2dn\Big]
\end{align}
where the inequality holds because conditioned on $h_{t_k}$, each column of $\Sigma^{-0.5}_{t_k} (\theta_1 -\tilde{\theta}_k)$ is the difference of two $d$-dimensional i.i.d. random vectors $\sim N(0,I)$.

As a result,
\begin{align}
& \ee \Big[\sum_{k=1}^{K_T}\sum_{t=t_k}^{t_{k+1}-1}  ||\Sigma^{-0.5}_t(\theta_1-\tilde{\theta}_k )||^2\Big]
\notag\\
\leq & 4dn\ee \Big[\mathds{1}_{\{t_{k}\leq T\}} (T_{k-1}+1)\Big]
\notag\\
\leq & 4dn \ee[T + K_T].
\end{align}
\end{proof}

\begin{lemma}
\label{lm:bounds_some_expectations}
The following bounds hold:
\begin{align}
&\ee[\sqrt{\log(X_T)}]\leq  \tilde O ( 1 )
\\
&\ee[\sqrt{\log(X_T)}X_T^2]
\leq \tilde O\Big( (1-\delta)^{-2}\Big),
\\
&\ee[X_T^4\log(X_T)]
\leq  \tilde O\Big((1-\delta)^{-5}  \Big).
\end{align}
\end{lemma}
\begin{proof}
Using Lemma \ref{lm:bound_X} on $X_T$, we get
\begin{align}
&\ee[\sqrt{\log(X_T)}]
\leq \sqrt{\ee[\log(X_T)]}
\leq \sqrt{\log(\ee[X_T])]}
\notag\\
\leq &O\Big(\sqrt{\log(\log(T)(1-\delta)^{-1})} \Big)
\notag\\
\leq &\tilde O(1).
\end{align}
Similarly,
\begin{align}
&\ee[\sqrt{\log(X_T)}X_T^2]
\leq \sqrt{\ee[\log(X_T)] \ee[X_T^4]}
\notag\\
\leq &\sqrt{ \log(\ee[X_T])]\ee[X_T^4]}
\notag\\
\leq &O\Big( (1-\delta)^{-2}\sqrt{ \log(T)\log(\log(T)(1-\delta)^{-1})} \Big)
\notag\\
\leq &\tilde O\Big( (1-\delta)^{-2} \Big).
\end{align}
Since $\log(X_T) \leq  X_T$, we have
\begin{align}
\ee[X_T^4\log(X_T)]
\leq &\ee[X_T^5]
\notag\\
\leq & O\Big(  \log(T)(1-\delta)^{-5} \Big)
\notag\\
\leq & \tilde O\Big( (1-\delta)^{-5}  \Big)
\end{align}
\end{proof}

\begin{IEEEbiographynophoto}{Yi Ouyang}(S'13-M'16)
received the B.S. degree in Electrical Engineering from the National Taiwan University, Taipei, Taiwan in 2009, and the M.S. degree and the Ph.D. degree in electrical engineering and computer science from the University of Michigan, Ann Arbor, MI, in 2012 and 2016, respectively. He was a Postdoctoral Scholar at the University of Southern California from 2016 to 2017. He is currently a Postdoctoral Scholar at the University of California, Berkeley.  His research interests include stochastic control, reinforcement learning, decentralized decision-making, and dynamic games with asymmetric information.
\end{IEEEbiographynophoto}

\begin{IEEEbiographynophoto}{Mukul Gagrani}
received the B.tech and M.tech degree in Electrical Engineering from the Indian Institute of Technology, Kanpur, India in 2013. 
He is currently a PhD candidate in the department of electrical engineering at the University of Southern California, Los Angeles, CA.
His research interests include decentralized stochastic control, stochastic scheduling and decision-making under uncertainty.
\end{IEEEbiographynophoto}

\begin{IEEEbiographynophoto}{Rahul Jain}
	is an associate professor and the K. C. Dahlberg Early Career Chair in the EE department at the University of Southern California. He received his PhD in EECS and an MA in Statistics from the University of California, Berkeley, his B.Tech from IIT Kanpur. He is winner of numerous awards including the NSF CAREER award, an IBM Faculty award and the ONR Young Investi gator award. His research interests span stochastic systems, statistical learning, queueing systems and game theory with applications in communication networks, power systems, transportation and healthcare.
\end{IEEEbiographynophoto}

\end{document}